\documentclass[12pt]{amsart}
\usepackage{amsmath,amssymb,color,natbib}

\usepackage{tikz}
\usetikzlibrary{snakes,calc,arrows}

\newcommand{\midarrow}{\tikz \draw[-triangle 90] (0,0) -- +(.1,0);}

\newtheorem{lemma}{Lemma}

\theoremstyle{theorem}

\newtheorem*{theoremp}{Theorem}

\theoremstyle{remark}

\newcommand{\df}[1]{\textit{#1}}

\def\mb{\mathbin}

\def\ep{\varepsilon}
\def\t{\theta}
\def\Re{\mathbf{R}}

\title{Testing for separability is hard}
\author{Federico Echenique}
\thanks{Division of the Humanities and Social Sciences, Caltech. 
The first version of this paper was dated August 2013; the current version
is from January 2014. I am
  grateful to Thomas Demuynck and John Quah for comments on the first
  draft of this note.}
\begin{document}
\maketitle

\begin{abstract} This paper shows that it is computationally hard to
  decide (or test) if a consumption data set is consistent with
  separable preferences.
\end{abstract}

\section{Introduction}

The assumption of separable preferences is ubiquitous in
economics. Economists assume separability of preferences, virtually
without ever testing this assumption empirically.  Here I argue
that there is a reason for such lack of empirical scrutiny: 
The problem of deciding if a data set
is consistent with separable preferences is computationally
hard. There cannot exist a test that is practical on large data sets,
and that serves to test for separability.

Every empirical study on consumption assumes, explicitly or
implicitly, that preferences are separable. For example, data on
supermarket purchases are used in isolation from other consumption
decisions. Or  data on consumption in one year is used without regard
for intertemporal consumption decisions. The choice among different
goods is analyzed while ignoring any consumption/leisure tradeoffs, and
independently of the allocation  of financial assets. All such
analyses, which depend on certain compartmentalizations in the
economy, rely on the assumption on separability. It is hard to
imagine a paper in applied economics that does not make some use of
separability. 

A few authors have proposed tests for separability. \cite{varian1983} has a
test that involves solving a system of polynomial
inequalities. \cite{cherchye2011} provide a computational approach to handling
Varian's system of inequalities. 
\cite{quah2012} has a test which is finite, meaning that one
would need to check if the data fit a finite number of
configurations. These tests are all hard to take to data because they
are computationally hard, and may be infeasible in large datasets.

My contribution here is to show that separability is inherently
hard. Specifically, that is is NP complete. This implies that it is as
hard as any problem in the class NP,  a class of problems that contains
all the natural decision problems studied in computer science.  NP
complete problems are widely regarded as intractable.

A result similar to mine has already appeared in
\cite{cherchye2011}. The main difference is that their result is asymptotic
in the number of goods as well as 
the number of observations. Cherchye et.\ al. 
 proved that separability is hard if one has a large data
set {\em and many goods}. In my view, it is important to establish the
result with a fixed number of goods because most studies in economics
use only a handful of goods. The construction used in the proof of the
theorem below uses 9 goods, the same number as in the classical study
on consumption by \cite{deaton1974analysis}.\footnote{This means that the
  problem is hard already with 9 goods. But the construction can probably be
  improved to use an even smaller number of goods.}  The number 9 is not a
limitation of 
old data sets and classical studies; recent studies on consumption
also use a small number 
of goods (for example \cite{cherchye2011} use 15 goods).  More broadly
speaking,  asymptotics on the size of the dataset simply seem  more
fundamental than on the number of goods.  

Finally, another difference with \cite{cherchye2011} is that they
focus on conditions for concave separability. My result is on
separability alone.

\section{Testing  separability}

We take consumption space to be $\Re^{n+m}_+$. The set of available
goods is partitioned in two, and we write a consumption bundle as
$x=(z,o)\in \Re^{n+m}_+$.  There are $n$ goods of  ``type $z$'' and $m$
of ``type $o$.''

A \df{data set} is a collection $(x_k,p_k)$, $k=1,\ldots,K$ in which for
every $k$ $x_k\in \Re^{n+m}_+$ is a consumption bundle
purchased at prices $p_k\in \Re^{n+m}_{++}$. 
Let $x_k = (z_k,o_k)\in\Re^n_+\times \Re^m_+$.

A data set $(x_k,p_k)$, $k=1,\ldots,K$ is \df{rationalizable by
  separable preferences} if there are monotone increasing  functions
  $u:\Re^n_+\rightarrow \Re$ and $v:\Re^{1+m} \rightarrow \Re$ such
  that \[ 
 v(u(z),o) <  v(u(z_k),o_k)
\] for all $(z,o)\in \Re^{n+m}_+$ with 
$p^z_k \cdot z + p^o_k\cdot o \leq
p^z_k \cdot z_k + p^o_k\cdot o_k$ and $(z,o)\neq (z_k,o_k)$.

\begin{theoremp}
The problem of deciding if a dataset has a separable rationalization
is NP-complete.
\end{theoremp}

\section{Proof}

\subsection{Notation and definitions}

By $e_i$ we denote  the $i$th unit vector in $\Re^n$, that is the
vector  that has zero in every 
entry except the $i$th, in which it has a one. Write $e_{12}$ for
$e_1+e_2$. The \df{embedding} of $\Re^n$ into $\Re^m$, with $n< m$ is the
function that maps  a vector $x\in \Re^n$ into the vector
$(x,0,\ldots,0)$ in $\Re^m$, which coincides with $x_i$ in the first
$n$ entries and then has a zero in the remaining entries. 

A \df{graph} is a set $X$ together with a binary relation $R\subseteq
X\times X$. We write $x\mb R y$ for $(x,y)\in R$. A sequence
$x_1,\ldots x_K$ in $X$ is a \df{path} from $x_1$ to $x_K$ if \[x_K\mb
R  x_{K-1}\mb R x_{K-2} \cdots x_2 \mb R x_1.\] A graph, or the
binary relation $R$, is \df{acyclic} if for every pair $x$ and $x'$,
with $x\neq x'$, 
if there is a path from $x$ to $x'$ then there is no path from $x'$ to
$x$.

\subsection{Construction}

We shall reduce from three-satisfiability. 
Consider a formula with $L$ clauses, $C_1,\ldots, C_L$, involving the
variables $x_1,\ldots,x_I$.

The strategy for the reduction is as follows. We introduce a pair of
bundles $z^1_i$ and $z^2_i$ for each variable $i$. These bundles are
not comparable by revealed preferences, but they are embedded into a
configuration of prices and bundles such that any rationalizing
preference must reflect an assignment of truth/falsehood to each
variable that makes all the clauses $C_l$ true. This is accomplished
using separability: in fact separability is crucial to make the
construction work with a fixed number of goods. The bundles $z^1_i$
and $z^2_i$ live in $\Re^2$ and they are shifted by adding different
amounts of the other goods so that they can play the same role in
different clauses. By separability, the comparison between $z^1_i$ and
$z^2_i$ must be the same in all shifted instances.

We first (Step 1) develop the
construction for a single clause. Then (Step 2) we tie the different
clauses together. The formal construction  follows. 

For each variable $x_i$, define the following vectors in $\Re_+^2$: 
$z^1_i = (i,1/i)$ and $z^2_i = (i+1/2,1/(i+1/2))$. 
Let $\tau:\{x_1,\ldots,x_I\}\rightarrow \{0,1\}$ 
be a truthtable for the variables $x_1,\ldots,x_I$. Define a binary
relation $B$ on $\{z^q_i:q=1,2; i=1,\ldots,I \}$ from $\tau$ by 
$z^1_i \mb B z^2_i$ if $\tau(x_i)=1$ and 
$z^2_i \mb B z^1_i$ if $\tau(x_i)=0$.

Note that each of the vectors in  $\{z^q_i:q=1,2; i=1,\ldots,K \}$ 
lie on the boundary of the convex set $A=\{ (\t_1,\t_2):
\t_2\geq 1/\t_1 \text{ and } \t_1>0\}$. Define $p^q_i$ to be such that
the hyperplane \[ 
\{(\t_1,\t_2) : p^q_i\cdot (\t_1,\t_2) = 1\}
\] supports $A$ at $z^q_i$. (A simple calculation
reveals that $p^1_i=(1/(2i),i/2)$ and $p^2_i=(1/(2i+1)),i/2+1/4)$, but
this does not play a role in the sequel.)

As a consequence of these definitions, we obtain the
following \begin{lemma}\label{lem:two} For all $i,i'=1,\ldots,I$, and all
  $q,q'=1,2$, we have  \[p^q_i\cdot z^q_i=1 < p^q_i\cdot z^{q'}_{i'}\] when
  $i'\neq i$ and/or $q'\neq q$.
\end{lemma}

To make the sequel easier to follow, we write $\rho(z^q_i)$ for
$p^q_i$. Define two positive numbers, $\ep$ and $M$ as follows.
 Let $\ep$ be such that 
\begin{equation}\label{eq:epsilon}
1 < \rho(z^q_i)\cdot z^q_i + \ep \rho(z^q_i)\cdot (1,1) <
\rho(z^q_i)\cdot z^{q'}_{i'}\end{equation} when 
  $i'\neq i$ and/or $q'\neq q$. In second place, let $M$ be such that 
\begin{equation}\label{eq:M}
\rho(z^q_i)\cdot z^{q'}_{i'} + 
\ep \rho(z^q_i)\cdot (1,1) < M\end{equation} when 
  $i'\neq i$ and/or $q'\neq q$.

\subsubsection{Step 1: The construction for a single clause.}

Consider a single clause $C$. Say that $C=y_i\vee y_j\vee y_h$, with
$y\in \{x,\bar x\}$.  Let 
$\hat z^q_i = z^q_i$ and $\hat p^q_i = p^q_i$ if $y_i = x_i$, and 
$\hat z^q_i = z^{3-q}_i$ and $\hat p^q_i = p^{3-q}_i$
if $y_i = \bar x_i$. Define $\hat z^q_j$, $\hat z^q_h$, $\hat p^q_j$,
and $\hat p^q_h$ analogously.

We are going to map the clause $C$ into a dataset in $\Re^6$. 

Consider the observations  
$\{(w_i,r_i):i=1,2,4,5,7,8\}$ and the set of bundles
$\{w_i:i=1,\ldots,9\}\subset \Re^8_+$ defined as follows:
\[ \begin{array}{ll}
w_1= \hat z^2_i + e_3 & r_1 = \rho(\hat z^2_i)+M e_3 + 2M(e_5+e_6) \\
w_2= \hat z^1_j+\ep(e_1+e_2) + e_4 & r_2 = \rho(\hat z^1_j) \\
w_3= \hat z^1_j + e_3 +e_7 & \\
w_4= \hat z^2_j + e_3 +e_7 & r_4= \rho(z^2_j) + Me_3 + 2M(e_4+e_6) \\
w_5= \hat z^1_h +\ep(e_1+e_2)+ e_5 & r_5=\rho(\hat z^1_h)\\
w_6= \hat z^1_h + e_3 +e_8 & \\
w_7= \hat z^2_h + e_3 +e_8 & r_7 = \rho(\hat z^2_h)+Me_3 + 2M(e_4+e_5)\\
w_8 = \hat z^1_i +\ep(e_1+e_2)+ e_6 & r_8=\rho(\hat z^1_i)\\
w_9 = \hat z^1_i + e_3 & \\
\end{array}\]
This means that bundles $w_i$ are purchased at prices $r_i$, for 
$i=1,2,4,5,7,8$. The bundles $w_3, w_6, w_9$ are added for
convenience. 

Let $X$ be the set of products taken from the first two, and the last
four, entries of the vectors $w_k$. That is, $X$ is the  
set of pairs $(z,o)$, with $z\in\Re^2$ and $o\in \Re^6$,
such that there is $z'$, $o'$, $w_k$ and $w_l$ 
with $w_k=(z,o')$ and $w_l=(z',o)$. Write $X_z$ for the projection of
$X$ onto $\Re^2$, and $X_o$ for the projection of
$X$ onto $\Re^6$; so $X = X_z\times X_o$.

The following tables contain the results of calculating $r_k\cdot
w_t$ (so that $r_k\cdot w_t$ is the content of the cell with row $r_k$
and column $w_t$). 


\[\begin{array}{l|l|l|l}
 & w_1 & w_2 & w_3    \\ \hline
r_1  & 1 + M 
& \rho(\hat z^2_i)\cdot \hat z^1_j +\ep \rho(\hat z^2_i)\cdot e_{12}
& \rho(\hat z^2_i)\cdot \hat z^1_j + M \\
r_2 
&\rho(\hat z^1_j) \cdot \hat z^2_i 
& 1 +\ep \rho(\hat z^1_j)\cdot e_{12}
& 1 \\
r_4
& \rho(z^2_j)\cdot \hat z^2_i + M 
& \rho(z^2_j)\cdot \hat z^1_j+\ep\rho(z^2_j)\cdot e_{12} +2M
& \rho(z^2_j)\cdot\hat z^1_j + M \\
r_5
& \rho(\hat z^1_h) \cdot \hat z^2_i 
& \rho(\hat z^1_h) \cdot \hat z^1_j+\ep \rho(\hat z^1_h) \cdot e_{12} 
& \rho(\hat z^1_h) \cdot \hat z^1_j  \\
r_7
& \rho(\hat z^2_h)\cdot \hat z^2_i + M
& \rho(\hat z^2_h)\cdot\hat z^1_j+\ep\rho(\hat z^2_h)\cdot e_{12} + 2M
& \rho(\hat z^2_h)\cdot\hat z^1_j + M \\
r_8
& \rho(\hat z^1_i)\cdot\hat z^2_i 
& \rho(\hat z^1_i)\cdot\hat z^1_j+\ep\rho(\hat z^1_i)\cdot e_{12} 
& \rho(\hat z^1_i)\cdot\hat z^1_j 
\end{array}\] 


\[\begin{array}{l|l|l}
  & w_4 & w_5   \\ \hline
r_1  
& \rho(\hat z^2_i)\cdot \hat z^2_j + M
& \rho(\hat z^2_i)\cdot \hat z^1_h +\ep\rho(\hat z^2_i)\cdot e_{12}+ 2M\\
r_2 
& \rho(\hat z^1_j)\cdot \hat z^2_j 
& \rho(\hat z^1_j)\cdot \hat z^1_h +\ep\rho(\hat z^1_j)\cdot e_{12}\\
r_4
& 1 + M
& \rho(z^2_j)\cdot\hat z^1_h +\ep\rho(z^2_j)\cdot e_{12}\\
r_5
& \rho(\hat z^1_h) \cdot \hat z^2_j 
& 1 +\ep\rho(\hat z^1_h) \cdot e_{12}\\
r_7
& \rho(\hat z^2_h)\cdot\hat z^2_j + M 
& \rho(\hat z^2_h)\cdot\hat z^1_h +\ep\rho(\hat z^2_h)\cdot e_{12}+
2M \\
r_8
& \rho(\hat z^1_i)\cdot\hat z^2_j 
& \rho(\hat z^1_i)\cdot\hat z^1_h +\ep\rho(\hat z^1_i)\cdot e_{12}
\end{array}\] 


\[\begin{array}{l|l|l|l|l}
 & w_6 & w_7 & w_8 & w_9  \\ \hline
r_1
& \rho(\hat z^2_i)\cdot \hat z^1_h + M
& \rho(\hat z^2_i)\cdot \hat z^2_h + M
& \rho(\hat z^2_i)\cdot \hat z^1_i +\ep\rho(\hat z^2_i)\cdot e_{12}+ 2M
& \rho(\hat z^2_i)\cdot \hat z^1_i +M \\
r_2
& \rho(\hat z^1_j)\cdot \hat z^1_h 
& \rho(\hat z^1_j)\cdot \hat z^2_h 
& \rho(\hat z^1_j)\cdot \hat z^1_i +\ep\rho(\hat z^1_j)\cdot e_{12}
& \rho(\hat z^1_j)\cdot \hat z^1_i \\
r_4
& \rho(z^2_j)\cdot\hat z^1_h + M
& \rho(z^2_j)\cdot \hat z^2_h + M
& \rho(z^2_j)\cdot \hat z^1_i +\ep\rho(z^2_j)\cdot e_{12}+ 2M
& \rho(z^2_j)\cdot \hat z^1_i + M \\
r_5
& 1
& \rho(\hat z^1_h) \cdot\hat z^2_h 
& \rho(\hat z^1_h) \cdot \hat z^1_i +\ep\rho(\hat z^1_h) \cdot e_{12}
& \rho(\hat z^1_h) \cdot \hat z^1_i \\
r_7
& \rho(\hat z^2_h)\cdot\hat z^1_h + M
& 1 + M
& \rho(\hat z^2_h)\cdot\hat z^1_i +\ep\rho(\hat z^2_h)\cdot e_{12}
& \rho(\hat z^2_h)\cdot\hat z^1_i + M \\
r_8
& \rho(\hat z^1_i)\cdot\hat z^1_h 
& \rho(\hat z^1_i)\cdot\hat z^2_h 
& 1 +\ep\rho(\hat z^1_i)\cdot e_{12}
& 1
\end{array}\]

Define the graph $(X,R)$ by letting $w \mb R w'$ if and only if $w=w_k$
for some $k$ and $r_k\cdot w_k > r_k\cdot w'$. Careful (if tedious)
inspection of the calculations above, (and using
inequalities~\eqref{eq:epsilon} and ~\eqref{eq:M}) reveal that 
that \begin{equation}\label{eq:defR} 
R = \{ (w_1,w_2), (w_2,w_3), (w_4,w_5), (w_5,w_6),
(w_7,w_8), (w_8,w_9)
\}.\end{equation} 

Let $\tau$ be a truth table for which our clause  $C=y_i\vee y_j\vee
y_h$ is true. That is: $\tau(y)=1$ for at least one $y\in \{y_i, y_j,
y_h\}$. Let $B$ be the binary relation induced by $\tau$ (see the
definition above). 

Note that $B$ is an acyclic binary relation on $\{\hat
z^q_k:q=1,2,k=i,j,h\}$. None of the vectors in $\{\hat
z^q_k:q=1,2,k=i,j,h\}$ is larger than the other (in the usual order on
$\Re^2$). There is therefore a function $u:\Re^2_+\rightarrow \Re$ for
which $u(z)>u(z')$ whenever $z\mb B z'$ or $z > z'$. 

Let $R'$ be defined by: $w\mb R' w'$ if (1)
$w\mb R w'$ or (2) if there are $z,z'\in X_Z$ and $o$ such that
$w=(z,o)$, $w'=(z',o)$ and $u(z) > u(z')$. Say that pair $(w,w')$ is a
1-edge if $w\mb R' w'$ for reason (1) and a 2-edge if $w\mb R' w'$
for reason (2). 

\begin{lemma}\label{lem:above}
 $(X,R')$ is acyclic.\end{lemma}

\begin{proof}
Suppose, towards a contradiction, that there is a cycle. The cycle
cannot consist purely of 2-edges because each 2-edge implies an
increase in $u(z)$. Therefore
some of the edges must consist be 1-edges. Inspection of the graph
$(X,R)$ reveals that all of the edges in $R$ must then be part of this
cycle. The reason is that edges can connect $(z,o)$ and $(z',o')$ with
$o\neq o'$ only if they belong to $R$. Then a cycle can only be closed
if it involves {\em all} of the edges in $R$. Such a cycle would define a path
from $w_9$ to $w_1$, from $w_3$ to $w_4$, and from $w_6$ to
$w_7$. Each of these paths would involve only 2-edges. By definition
of 2-edges, then $u(w_1) > u(w_9)$,
$u(w_4) > u(w_3)$, and $u(w_7) > u(w_6)$. But $u(w_1) > u(w_9)$ can
only be true if $\hat z^2_i \mb B \hat z^1_i $. Similarly, we obtain
that $\hat z^2_j \mb B \hat z^1_j $ and $\hat z^2_h \mb B \hat z^1_h
$. This contradicts that $C$ is true under the truthtable $\tau$.
\end{proof}



\begin{lemma}\label{lem:three}
There is a function $v$ such that $v(u(z),o) > v(u(z'),o')$ whenever 
$(z,o) \mb R' (z',o')$ or $u(z)\geq u(z')$ and $o > o'$.
\end{lemma}

\begin{proof}
Let $R''$ be a binary relation on $u(\Re^2_+)\times \Re^4_+$ defined by 
$(s,o)\mathbin{R''} (s',o')$ if $(s,o)> (s',o')$ or if there is $z$
and $z'$ such that $s=u(z)$ and $s'=u(z')$ and 
$(z,o)\mathbin{R} (z,o)$. By Lemma~\ref{lem:above} and the observation
that non of the vectors in $X_o$ is comparable in the usual Euclidean
order, the relation $R''$ is
acyclic. The set $u(X_z)\times X_o \cup \mathbb{Q}\cap u(\Re^2_+)
\times \mathbb{Q}^4_+$ is countable and order dense, so there is a
function $v$ as required by the statement of  the lemma.
\end{proof}

\subsubsection{Step 2: The construction for $L$ clauses.}

For each clause $C_l$, define the bundles $w_t$ and prices $r_k$ as
above. Let $w^l_t$ be defined as the vector in $\Re^9$ obtained by
embedding $w_t$ and adding $M2^l e_9$. Let $r^l_k$ be the sum of $e_9$
and the  embedding of $r_k$ in $\Re^9$. As a result of these
definitions, $r^l_k\cdot w^k_t = r_k\cdot w_t + M2^l$. Importantly,
for a fixed $l$, the comparison of  
$r^l_k\cdot w^k_k$ and $r^l_k\cdot w^k_t$ is the same as the
comparison of $r_k\cdot w_k$ and $r_k\cdot w_t$ performed in Step~1. 

Define $X^l=X^l_z\times X^l_o$ from $w^l_1,\ldots,w^l_9$ in the same
way as $X=X_z\times X_o$ was 
defined. Define $R^l$ from $w^l_t$ and $r^l_k$ in the same way as $R$
was defined. Note that $r^l_k\cdot w^k_t$ only
differ from $r_k\cdot w_t$ in the constant $M2^l$, so the graphs 
$(X,R)$ and $(X^l,R^l)$  are the same once we identify $w_k$ with $w^l_k$. 

Let $\bar X = \cup_{l=1}^L X^l$. Define a binary relation $\bar R$ on
$\bar X$ by $w \mathbin{\bar R} w'$ iff there is some $l$ and some $t$ such
that $w=w^l_t$ and $r^l_t\cdot w^l_t > r^l_t\cdot w'$. Note
that $\bar R$ has the following properties \begin{enumerate}
\item 
 $\bar R$ coincides with $R^l$ on $X^l$ ($R^l = \bar R\cap X^l\times X^l$);
\item if $w\in X^l$ and  $w'\in X^{l'}$, with  $l'< l$, 
then $w \mathbin{\bar R} w'$ and it is false that $w' \mathbin{\bar R}
w$. 
\end{enumerate}

Let $\tau$ be a truthtable for which all clauses $C^l$ are true. Such
a truth table defines a binary relation $B$ on $\bar X_z = \cup_l
X^l_z$. The binary relation is acyclic, as there are no pairs of
subsequent edges in $B$. As in Step 1, there is a function
$u:\Re^2_+\rightarrow \Re$ for which $u(z)>u(z')$ whenever $z\mb B z'$
or $z > z'$.  


Let $\bar R'$ be defined by: $w\mb R' w'$ if (1)
$w\mb R w'$ or (2) if there are $z,z'\in \bar X_z$ and $o$ such that
$w=(z,o)$, $w'=(z',o)$ and $u(z) > u(z')$. 


\begin{lemma}\label{lem:four}
$(\bar X,\bar R')$ is acyclic.
\end{lemma}

\begin{proof}
By the second property of $\bar R$, there cannot exist a cycle that
contains an edge going from $w\in X^l$ to
$w'\in X^{l'}$ with $l\neq l'$. Therefore any cycle must contain only
vertexes in some $X^l$. By Lemma~\ref{lem:two}, there is no such
cycle. 
\end{proof}


The following result, which finishes the proof,  
follows from Lemma~\ref{lem:four} in a similar
way to how Lemma~\ref{lem:three} follows from Lemma~\ref{lem:two}.

\begin{lemma}\label{lem:five}
There is a function $v$ such that $v(u(z),o) > v(u(z'),o')$ whenever 
$(z,o) \mathbin{\bar R'} (z',o')$ or $u(z)\geq u(z')$ and $o > o'$.
\end{lemma}

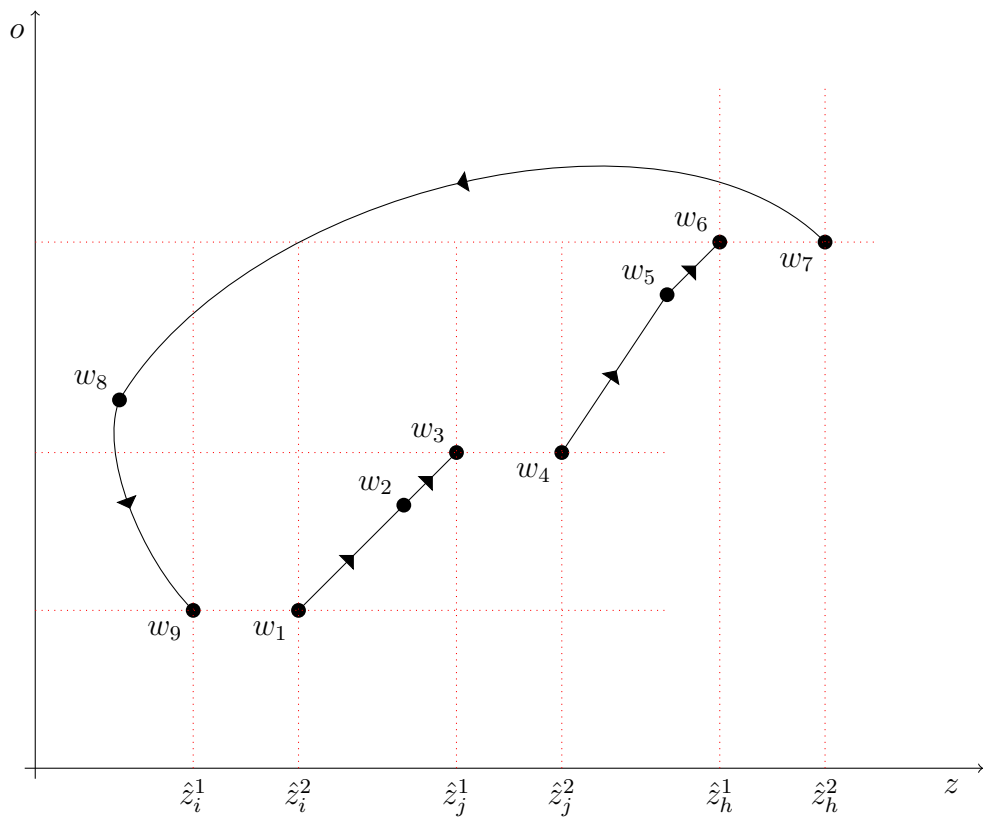
\begin{figure}[hbbbb]
\begin{tikzpicture}[scale=1.4]
\draw[->] (0,-.1) -- (0,7.2);
\draw[->] (-.1,0) -- (9,0);

\path (8.7,0) node[anchor=north] {$z$} ;
\path (0,7) node[anchor=east] {$o$} ;

\fill (1.5, 1.5) circle (2pt); \fill (2.5, 1.5) circle (2pt);
\path (1.5, 1.5) node[anchor=north east] {$w_9$} 
      (2.5,1.5) node[anchor=north east] {$w_1$}; 
\draw[-,dotted,red] (1.5,0) -- (1.5,5) ;
\path (1.5,0) node[anchor=north] {$\hat z^1_i$} ;
\draw[-,dotted,red] (2.5,0) -- (2.5,5) ;
\path (2.5,0) node[anchor=north] {$\hat z^2_i$} ;
\draw[-,dotted,red] (0,1.5) -- (6,1.5) ;

\begin{scope}[every node/.style={sloped,allow upside down}]
\draw[-] (2.5, 1.5) -- node {\midarrow}  (3.5, 2.5);
\draw[-] (5,3) -- node {\midarrow}  (6,4.5);
\draw[-] (6, 4.5) -- node {\midarrow}  (6.5 ,5) ;
\draw[->] (7.5, 5) .. controls  (6, 6.5) and (2, 5.5) ..  node {\midarrow} (.8, 3.5);
\draw[-] (.8, 3.5) .. controls  (.6, 3) and (1, 2) ..  node
{\midarrow} (1.5, 1.5);

\draw[-] (3.5, 2.5) -- node {\midarrow} (4,3) ;
\end{scope}

\fill (3.5, 2.5) circle (2pt);
\fill (4, 3) circle (2pt); \fill (5, 3) circle (2pt);
\path (3.5, 2.5) node[anchor=south east] {$w_2$} 
      (4,3) node[anchor=south east] {$w_3$}
      (5,3) node[anchor=north east] {$w_4$};

\draw[-,dotted,red] (4,0) -- (4,5) ;
\path (4,0) node[anchor=north] {$\hat z^1_j$} ;
\draw[-,dotted,red] (5,0) -- (5,5) ;
\path (5,0) node[anchor=north] {$\hat z^2_j$} ;

\draw[-,dotted,red] (0,3) -- (6,3) ;

\fill (6, 4.5) circle (2pt);
\fill (6.5, 5) circle (2pt); \fill (7.5, 5) circle (2pt);
\path (6, 4.5) node[anchor=south east] {$w_5$} 
      (6.5 ,5) node[anchor=south east] {$w_6$}
      (7.5, 5) node[anchor=north east] {$w_7$};
\draw[-,dotted,red] (6.5,0) -- (6.5,6.5) ;
\path (6.5,0) node[anchor=north] {$\hat z^1_h$} ;
\draw[-,dotted,red] (7.5,0) -- (7.5,6.5) ;
\path (7.5,0) node[anchor=north] {$\hat z^2_h$} ;
\draw[-,dotted,red] (0,5) -- (8,5) ;

\fill (.8, 3.5) circle (2pt);
\path (.8, 3.5) node[anchor=south east] {$w_8$} ;


\end{tikzpicture}
\caption{The construction for a single clause.}\label{fig:one}
\end{figure}

\subsection{Some remarks on the proof}

1) The construction for one clause can be summarized in the  diagram depicted in
Figure~\ref{fig:one}. The horizontal axis represents $\Re^2$ and the
vertical axis $\Re^6$.
The directions of the arrows reflect the binary
relation $R$: $w_1\mb R w_2$ is denoted by the arrow pointing from
$w_1$ to $w_2$, and so on.
Note that the pairs of bundles $w_9$ and $w_1$,
$w_3$ and $w_4$, and $w_6$ and $w_7$ share their $o$ component. These
pairs are the only ones that share an $o$ component.

2) The presence of the bundles $w_2$, $w_5$ and $w_8$ may need an explanation. We
need to use them for the following reason. Consider the case of
$w_2$. We want to have $w_1\mb R w_3$, but not that $w_1\mb R
w_4$. This is difficult because $w_3$ and $w_4$ differ only in the $z$
component. By introducing $w_2$, which dominates $w_3$ but not $w_4$,
we can achieve the desired relations.

3) I have taken a shortcut in the proof by introducing
exponential quantities $w^l_t$. They are there to make sure that
certain quantities are large enough, and are easily avoided.

4) The definition of the bundles $z^q_k$ and
supporting prices $\rho(z^q_i)$ may involve using irrational
numbers, which is questionable from al algorithmic viewpoint. Since the
inequalities in Lemma~\ref{lem:two} are strict,
these numbers can be replaced with rational numbers to have the
construction only operate with ``discrete'' objects. (In a similar
fashion, the primitive dataset should only involve consumption bundles and
prices with rational entries.)

5) The main contribution here is to do the construction for a fixed
number of goods. 
If one is free to use any number of different
goods to capture the different edges, then it is easy to recreate any
given graph as a revealed preference binary relation. When the number
of goods is fixed, not all graphs can be revealed preference relations (the best
known example is the case when there are two goods, in which the weak
axiom of revealed preference suffices for rationalizability). So it is
important to be able to work with a rather specific graph, the one
depicted in Figure~\ref{fig:one}. The ability to do the reduction for
a fixed number of goods relies, among other things, on using 3SAT.

\bibliographystyle{econometrica}
\bibliography{separability}

\end{document}